\documentclass[journal,comsoc]{IEEEtran}
\usepackage[T1]{fontenc}% optional T1 font encoding

\usepackage{graphicx,epsfig}
\usepackage[noadjust]{cite}
\usepackage{mcite}
\usepackage{amsfonts,helvet}
\usepackage{fancyhdr}
\usepackage{threeparttable}
\usepackage{epsf,epsfig}
\usepackage{amsthm}
\usepackage{amsmath}
\usepackage{siunitx}
\usepackage{amssymb}
\usepackage{dsfont}
\usepackage{subfigure}
\usepackage{color}
\usepackage{algorithm}
\usepackage{algpseudocode}
\usepackage{algcompatible}
\usepackage{enumerate}
\usepackage{hyperref}
\usepackage{cancel}
\newtheorem{theorem}{Theorem}

\newtheorem{proposition}{Proposition}

\newmuskip\pFqmuskip

\IEEEoverridecommandlockouts

\title{Spectral Efficiency Bounds for Interference-Limited SVD-MIMO Cellular Communication Systems}
\author{Jinseok Choi, Jeonghun Park, and Brian L. Evans \thanks{The authors are with the Wireless Networking and Communication Group
(WNCG), Department of Electrical and Computer Engineering, The University
of Texas at Austin, Austin, TX 78701, USA. (e-mail: \{jinseokchoi89, jeonghun\}@utexas.edu, bevans@ece.utexas.edu)
} \\
}

\begin{document}
\maketitle
\begin{abstract}
The ergodic spectral efficiency (SE) in interference-limited multiple-input multiple-output (MIMO) downlink cellular systems is characterized based on stochastic geometry.
% in a tractable network model.
A single user is served by using singular value decomposition precoding and combining.
%We put forth two assumptions for analytical tractability: each user and base station is distributed as a homogeneous Poisson point process, and 
%Exploiting the \textit{quarter circle law} and 2-D homogeneous Poisson point process, we approximate the expectation of eigenvalues and derive the upper and lower bounds of the spectral efficiency.
%To character the gain from SVD technique, we adopt the quarter circle law. 
%To analyze the channel gain of SVD technique, 
%We model base stations' and users' locations by using a homogeneous Poisson point process. 
By approximating the expectations of the channel eigenvalues, we derive upper and lower bounds on the ergodic SE.
% as functions of relevant system parameters, chiefly the number of antennas, the number of streams, and the pathloss exponent.
The obtained upper bound is the best possible system-level performance of any MIMO strategy in non-cooperative cellular networks.
%One noticeable point of our expressions is that they are not functions of particular signal-to-noise ratio (SNR), so as to serve as a useful system-level performance indicator averaging over possible user locations.
%Showing the trade-off between the multiplexing gain and the power gain, the obtained bounds provide the insights for the optimal number of streams maximizing the spectral efficiency.
%and numerical results demonstrate that the approximations tightly match with simulation results. The derived bounds provide close-fitting bounds with the average gap of 0.24 bps/Hz and illustrate that 
%Leveraging the bounds, we further estimate the optimal number of streams.
We validate our analytical results through simulation.
%Our major contribution is to obtain the lower and upper bounds as a function of system parameters: the number of antennas and non-zero transmit symbol, and the pathloss exponent. 
%Noticeable point is that our results is , so that 
We also conjecture that there exists the optimal number of streams being proportional to the pathloss exponent.
%Our major finding is that there exists the optimal number of streams that maximizes the spectral efficiency being proportional to the pathloss exponent.
\end{abstract}
\begin{IEEEkeywords}
MIMO, downlink, SVD, ergodic spectral efficiency, the quarter circle law, Poisson point process.
\end{IEEEkeywords}

\section{Introduction}\label{sec:intro}
%Justification-why this work
% (Importance of System-level analysis)
%Understanding the benefits of multiple-input multiple-output (MIMO) wireless techniques, they have been extensively studied to develop capacity-achieving transceiver technologies~\cite{gesbert2003theory, goldsmith2003capacity}. 
%Early works analyzed the performance of MIMO channel systems in terms of the signal-to-noise ratio (SNR)~\cite{paulraj2004overview}, and a SNR-based simulation study was universal. 
%With sufficient link-level analysis results, these works have given design guidelines to network providers. 
%The deterministic SNR performance analysis, however, cannot be a robust guideline as mobiles and associated base stations (BSs) do not have a fixed SNR in practice. 
%From this perspective, understanding the system-level performance of MIMO communications in a cellular network has become highly important.
%
% (What I'm going to do)
%There has been several approaches to the system-level analysis of MIMO wireless networks using 2-D homogeneous Poisson point process (PPP)~\cite{chiu2013stochastic} to conduct a stochastic performance analysis. 
%Understanding the benefits of multiple-input multiple-output (MIMO) techniques, 
\IEEEPARstart{T}{he} performance of MIMO systems has been extensively investigated over a few decades~\cite{tse2005fundamentals, paulraj2004overview, goldsmith2003capacity}.
% to develop capacity-achieving transceivers.
The prior work analyzed the MIMO systems in terms of the deterministic signal-to-noise ratio (SNR) and a SNR-based simulation study was universal~\cite{paulraj2004overview}. 
%With sufficient link-level analysis, this work has given network design guidelines.
%The results presented in [19]–[22] only holds for specific user locations, making it difficult to provide a system level analysis over many user locations and system parameters.
Under this approach, the results only hold for particular user locations, having difficulty to provide a system-level view incorporating many possible user locations.
%mobiles and associated base stations (BSs) do not have a fixed SNR in practice. 
%As a result, the system-level performance has been drawing more attention.

%Despite decades of research, however, the system-level performance analysis on the optimal spectral efficiency of MIMO channels in cellular networks has been challenging. 
%The major difficulties have been the characterization of eigenvalues in MIMO channels due to the exceptionally complicated probability density functions (PDF) of eigenvalues as well as the modeling of the inter-cell interference (ICI).
%In this paper, we utilize two key techniques to resolve the challenges: (i) the \textit{quarter circle law}~\cite{marvcenko1967distribution} to characterize eigenvalues, and (ii) 2-D homogeneous Poisson point process (PPP)~\cite{chiu2013stochastic} to model the MIMO system with channel state information at transmitters (CSIT) in full ICI.
 
%MIMO PPP related work 
%Modeling a network by using a PPP is a useful for extracting 
%Considering the network in which transmitters are distributed according to PPP, the performance of the wireless communication networks has been actively investigated.
To resolve the limitation, the system-level analysis has been performed in a tractable network model where the locations of the base stations (BSs) are modeled by using a homogeneous Poisson point process (PPP).
%under the assumption that users' and base stations' (BSs) locations are modeled by a homogeneous Poisson point process (PPP).
%1) Capacity in Ad hoc 
%In wireless ad hoc networks, the throughput gain can be linearly scaled with the number of receiver antennas using a partial zero-forcing (ZF) receiver~\cite{jindal2009rethinking, jindal2011multi}.  
%Transmission capacity is used as an ad hoc MIMO network performance measure in many papers~\cite{weber2010overview}.% Closed-form performance bounds in the networks have been derived in~\cite{weber2005transmission}. 
%~Although the results presented in~\cite{weber2010overview} gave clear insights for the design of ad hoc networks, they did not capture the multi-hop case of the network. 
%Using transmit stream adaptation and interference cancellation, upper and lower bounds on the transmission capacity of the ad hoc MIMO network with and without CSIT were derived in~\cite{vaze2012transmission}. 
%Under no CSIT assumption, the transmitter uses spatial multiplexing and the receiver performs partial ZF It is concluded that single stream beamforming is optimal with the worst case scenario, but this does not accurately reflect multiplexing gain. 
%2) Recently Random Cellular Network, tractable,  not grid simulation.
%In prior work,
{ Leveraging this Poisson network model, the signal-to-interference-plus-noise ratio was characterized in a single-input single-output downlink cellular system
%A tractable model for the downlink coverage probability and the achievable rate was developed with single-input single-output (SISO) channels~
\cite{andrews2011tractable,di2013average}. }
%The proposed model provided better insights than grid-based models, and the closed-forms of the coverage probability and the achievable rate were derived.
%, and especially in Rayleigh fading channels, they become simple functions of deployment density $\lambda$.
%As extensions, the prior work \cite{??}  
%while they relied on ZF mainly due to its analytical tractability.
{ Extending \cite{andrews2011tractable} to the MIMO broadcast channel \cite{veetil2013coverage}, the coverage probability and rate were derived by using a zero-forcing (ZF) receiver with inter cell interference (ICI) cancellation. 
Further in \cite{veetil2015performance}, the optimal number of antennas for ICI cancellation when using partial ZF receivers was derived along with analytical expressions for coverage and rate distribution.}
In \cite{dhillon2013downlink}, the performance of the downlink multi-antenna heterogeneous cellular network was studied for ZF precoding.
% and it is demonstrated that for a given total number of transmit antennas, spreading them across multiple single-antenna BSs rather than fewer multi-antenna BSs is more preferable in terms of coverage and rate per user.
Cooperation gain under ZF beamforming \cite{hosseini2016stochastic} were investigated for MIMO networks, which provides the optimal loading factor ($\simeq 0.6$) to maximize per-BS ergodic sum rate.
{ Formulating ICI as an infinite sum of independent and conditionally distributed Gaussian random variables, average symbol error probability formulas for MIMO cellular networks were derived \cite{lu2015stochastic, di2015stochastic}. }
%The ergodic sum rate with respect to a loading factor was investigated when disjoint BS clusters were used with ZF beamforming to mitigate the inter-cell interference (ICI) in \cite{hosseini2016stochastic}.

%\begin{figure}[t]\centering
%\includegraphics[scale = 0.24]{SVD_MIMO.png}
%\caption{The point-to-point single-carrier SVD-MIMO downlink system. The matrix $\mathbf{V}$ and $\mathbf{U}$ are from SVD of the channel matrix $\mathbf{H = U\Lambda V}^{H}$.} \label{fig:1}
%\end{figure}

A common limitation in prior work \cite{veetil2013coverage,veetil2015performance, dhillon2013downlink, lu2015stochastic, hosseini2016stochastic} is the use of ZF precoder and/or equalizer due to analytical tractability. 
In a MIMO system, however, the ZF is suboptimal as it cannot extract the power gain from channels. 
In a single-user MIMO system, the optimal strategy is singular value decomposition (SVD) based precoding and combining. 
Thus, we analyze the system-level performance of an SVD-MIMO system. 
%to provide a system-level view 
%regarding the optimal tranceiver that extracts the full MIMO gain.
%The operating SIR regime is low, say its expected value is less than 
%($2$ bits/sec/Hz~\cite{andrews2011tractable})
%This motivates to use singular-value decomposition (SVD) precoder/receiver to extract the full MIMO gain in a cellular network.

% Limitations of Previous Work 1
In this paper, we characterize the ergodic spectral efficiency (SE) of MIMO downlink cellular systems in which a single user with multiple antennas is served by using SVD precoding and combining.
%{\bf{More clear idea and contribution (also major finding)}}
%The system-level performance analysis of SVD-MIMO channels in cellular networks has been challenging. 
The major difficulty of the analysis has been the characterization of eigenvalues and the ICI, coupled with each other in MIMO channels. 
%This is because the eigenvalues have the exceptionally complicated probability density functions (PDF)~\cite{wigner1965distribution} 
% since the channel gain is represented in an eigenvalue form in the SVD-MIMO.
To resolve such challenge, we use two key techniques: (i) the \textit{quarter circle law}~\cite{marvcenko1967distribution} to characterize eigenvalues, and (ii) 2-D homogeneous PPP~\cite{chiu2013stochastic} to model the MIMO networks. 
Leveraging these techniques, we approximate the expectation of each eigenvalue and derive the bounds of the ergodic SE of SVD-MIMO channels by decoupling the eigenvalue power gain from the ICI. 
%We also validate the derived expressions by simulations. 
{ Since we use the optimal precoder and combiner, our result serves as an upper bound on the system-level performance of any MIMO strategy in a non-cooperative cellular network under the assumptions of equal power allocation and no interference cancellation. }
This is not the case in the prior work~\cite{veetil2013coverage, veetil2015performance, dhillon2013downlink, lu2015stochastic, hosseini2016stochastic,di2015stochastic}.
%{\bf{Major finding}}
%As a byproduct, 
We also observe that there exists the optimal number of streams that maximizes the ergodic SE and it is proportional to the pathloss exponent.
%illustrate that there exists the optimal number of streams maximizing the SE and the number of streams is proportional to the pathloss exponent of the network.
%The major finding is that unlike the previous work~\cite{veetil2013coverage, dhillon2013downlink} that analyzed the ZF-MIMO performance, the obtained bounds provide the insights for the optimal MIMO system performance. 

{\it Notation}: $\bf{A}$ is a matrix and $\bf{a}$ is a column vector. $\mathbf{A}^{H}$, $\mathbf{A}^\intercal$ and $\mathbf{A}^{-1}$ denote conjugate transpose, transpose and inverse, respectively. 
$\mathbf{A}^{(n)}$ represents a first ($n\times n$) sub-matrix, and $\mathbf{a}^{(n)}$ is a first ($n\times 1$) sub-vector $\--$ ($1\times n$) for a row vector. $\mathcal{CN}(\mu, \sigma^2)$ is a complex Gaussian distribution with mean $\mu$ and variance $\sigma^2$. 

\section{System Model}\label{sec:sys_model}

%\subsection{Network Model}
We consider a single-user downlink cellular network model in which BSs are distributed according to a homogeneous PPP,  $\Phi = \{{\bf{d}}_i ,i\in \mathbb{N}\}$ of intensity $\lambda$. Users are also distributed as an independent homogeneous PPP, $\Phi_{\rm U} = \{{\bf{u}}_i ,i\in \mathbb{N}\}$.
Each BS's coverage area is presented as a Voronoi cell yielding minimum pathloss. Each BS and user are equipped with $N$ antennas.

%\subsection{Signal Model}
We assume that each BS transmits $m\leq N$ streams to its associated user with equal power allocation.
% where the power is equally allocated to each stream.
%The pathloss $r^{-\alpha}$ is used where $r$ denotes the distance between the typical user and its BS, and $\alpha$ is a pathloss exponent with $\alpha > 2$. 
%We assume IID Rayleigh fading channels, i.e., all channel coefficients follow complex Gaussian distribution $\mathcal{CN}(0, 1)$, and SVD-MIMO is employed with CSIT.
%An independent Rayleigh fading channel is assumed, i.e., each entry of the channel matrix follows an independent and identically (\textit{i.i.d.}) zero mean and unit variance. 
Denoting $\mathbf{H}_i \in \mathbb{C}^{N\times N}$ as the channel matrix between the user at $\mathbf{u}_i$ and the associated BS at $\mathbf{d}_i$, its SVD is represented as $\mathbf{H}_{i} = \mathbf{U}_{i} \mathbf{\Lambda}_{i}\mathbf{V}_{i}^H$. 
The matrices $\mathbf{U}_i$ and $\mathbf{V}_i$ are $N \times N$ unitary matrices and $\mathbf{\Lambda}_i \in \mathbb{R}^{N \times N}$ is a diagonal matrix of singular values $\sigma_k$;  $\mathbf{\Lambda}_i = \text{diag}(\sigma_1,\cdots,\sigma_N)$ where $\sigma_1 \geq \cdots \geq \sigma_N$.
%$\mathbf{\Lambda}_i = \text{diag}(\sigma_{1}, \sigma_{2}, \cdots, \sigma_{N}) $. 
%Each element of the channel matrix $[\mathbf{H}_{0}]_{i,j}$ represents the channel between the $j$th antenna of $b_{0}$ and the $i$th antenna of the user.  
%It transmits $m$ symbols and $N-m$ zeros, 
Assuming the perfect channel state information at transmitters and at receivers, the BS at $\mathbf{d}_i$ transmits symbols $\mathbf{s}_{i} = [s_{i,1}, s_{i,2}, \cdots, s_{i,m}, 0, \cdots,0]^{\intercal} \in \mathbb{C}^{N\times 1}$, with $\mathbb{E}[|s_{i,k}|^2] = 1/m$ to its associated user at $\mathbf{u}_i$ through a precoding matrix $\mathbf{V}_{i} \in \mathbb{C}^{N\times N}$.
Then, the received signals go through $\mathbf{U}_{i}^{H} \in \mathbb{C}^{N\times N}$.
%ICI cancellation is not considered in this paper.
%as it is negligible in the cell interior due to the high signal-to-noise ratio (SNR), and also at the cell edge due to much larger interference power than the noise power in most modern cellular networks. 

Per Slivnyak's theorem~\cite{baccelli2009stochastic}, we consider a typical mobile user at the origin $\mathbf{u}_1 = \mathbf{0}$.
Noting that $\mathbf{s}_1$ has $m$ non-zero entries, the received signal $\mathbf{y}^{(m)} \in \mathbb{C}^{m\times 1}$ is given by
%\vspace{-0.5em}
%
% EQUAITON #1
%\small
\begin{equation}\label{eq:1}
\mathbf{y}^{(m)} = \|\mathbf{d}_1\|^{-\frac{\alpha}{2}}\mathbf{\Lambda}^{(m)}_{1}\,\mathbf{s}_{1}^{(m)}+\sum_{i = 2}^{\infty}{{\|\mathbf{d}_{i}\|}^{-\frac{\alpha}{2}}(\mathbf{\tilde{H}}^H_{i})^{(m)}\mathbf{s}_{i}^{(m)}} + \tilde{\mathbf{n}}^{(m)}
\end{equation}
\normalsize
%
% EQUAITON #1
%\begin{equation}\label{eq:1}
%\mathbf{y}^{(m)} = r^{-\alpha/2}\mathbf{\Lambda}^{(m)}_{0}\,\mathbf{s}_{0}^{(m)}+\underbrace{\sum_{l\in \Phi\setminus b_{0}}{R_{l}^{-\alpha/2}[\mathbf{\check{H}}^H_{l}]^{(m)}\mathbf{s}_{l}^{(m)}}}_{I(\Phi)}
%\end{equation}
%
with $\mathbf{\tilde{H}}_{i}^{H} = \mathbf{U}_{1}^{H}\mathbf{H}_{i,1}\mathbf{V}_{i}$, where $\mathbf{H}_{i,1}$ is the channel matrix between the BS at $\mathbf{d}_i$ and the typical user. $\mathbf{\Lambda}_1 = \text{diag}(\sigma_{1}, \cdots, \sigma_{N}) $ and $\tilde{\mathbf{n}}= \mathbf{U}_{1}^{H}\mathbf{n}$ where $\mathbf{n}\sim \mathcal{CN}(\mathbf{0},\nu\mathbf{I}_{\rm N})$ is the additive white Gaussian noise. The pathloss exponent is considered as $\alpha > 2$. 
We assume Rayleigh fading, i.e., all the channel coefficients follow the IID $\mathcal{CN}(0, 1)$. 
Since $\mathbf{U}_{1}^{H}$ and $\mathbf{V}_{i}$ are unitary,
%and $\mathbf{V}_i$ is not from the SVD of $\mathbf{H}^{(1)}_i$, 
the matrix $\mathbf{\tilde{H}}_{i}$ is also an IID Rayleigh fading channel matrix and $\mathbf{\tilde{n}}\sim \mathcal{CN}(\mathbf{0},\nu\mathbf{I}_{\rm N})$.%The second term in~\eqref{eq:1} is the ICI from the other BSs with the distance $\mathbf{d}_i$ from the typical user. 
%The vector $\mathbf{s}_{l} \in \mathbb{C}^{N \times 1}$ is the transmit signals from $b_l$. 

From~\eqref{eq:1}, the $k$th received signal becomes
\vspace{-0.3 em}
% EQUAITON #2
\begin{equation}\label{eq:2}
y_k = \|\mathbf{d}_1\|^{-\frac{\alpha}{2}}{\sigma}_{k}\,{s}_{1,k}+\sum_{i=2}^{\infty}{\|\mathbf{d}_{i}\|^{-\frac{\alpha}{2}}(\mathbf{\tilde{h}}_{i,k}^{H})^{(m)}\,\mathbf{s}^{(m)}_{i}} + \tilde{n}_k
\end{equation}
where $\mathbf{\tilde{h}}_{i,k}^{H}$ is the $k$th row vector of $\mathbf{\tilde{H}}_{i}^{H}$. 
We assume that the noise is negligible~\cite{andrews2011tractable}.
The signal-to-interference ratio (SIR) of $y_k$ is expressed as
\vspace{-0.3 em}
% EQUAITON #3
\begin{equation}\label{eq:3}
\text{SIR}_k = \frac{\|\mathbf{d}_1\|^{-\alpha}{\sigma}^{2}_{k}}{\sum_{i=2}^{\infty}{\|\mathbf{d}_i\|^{-\alpha}{q}_{i,k}}}
\end{equation}
where $q_{i,k} = \left|\big(\mathbf{\tilde{h}}_{i,k}^{H}\big)^{(m)}\mathbf{\tilde{h}}_{i,k}^{(m)}\right|$. Since $q_{i,k}$ is the sum of $m$ exponential random variables, $q_{i,k}$ follows Chi-squared distribution with $2m$ degree-of-freedom $\chi^{2}_{2m}$.
%================================================================================%
%================================================================================%

\section{Main Results}\label{sec:main}

In this section, we derive the upper and lower bounds of the ergodic SE for~\eqref{eq:1} without the noise. 
The ergodic SE of the typical user is expressed as
\vspace{-0.3 em}
% EQUAITON #4
\begin{equation}\label{eq:4}
r(N, m, \alpha, \lambda) = \mathbb{E}\bigg[ \sum_{k=1}^{m}{\log_2(1+\text{SIR}_k)}\bigg].
\end{equation}

\subsection{Eigenvalue Characterization}\label{subsec:A}
To characterize eigenvalues, we exploit the asymptotic distribution of eigenvalues instead of the non-asymptotic distribution due to its exceptionally high complexity~\cite{wigner1965distribution}.
%to find the upper and lower bounds of the ergodic SE~\eqref{eq:4}. 
From the~\textit{quarter circle law}, the PDF and CDF of the eigenvalue $X$ of $\mathbf{H}/\sqrt{N}$~\cite{tse2005fundamentals}, where $\bf  H$ is the $N\times N$ channel matrix whose entries are distributed as $\mathcal{CN}(0, 1)$, are given as

% EQUAITON #5
\small
\begin{equation}\label{eq:5}
f_{X}(x) = \frac{1}{\pi}\sqrt{\frac{1}{x}-\frac{1}{4}},    \quad \quad \text{for } 0< x \leq 4
\end{equation}
\normalsize
%and its CDF becomes
%
% EQUAITON #6
\small
\begin{equation}\label{eq:6}
F_X(x) = \frac{1}{2\pi}\left\{\pi + x\sqrt{\frac{4}{x}-1}-2\tan^{-1}\left(\frac{x-2}{x-4}\sqrt{\frac{4}{x}-1}\right)\right\},
\end{equation}
\normalsize
and the PDF and CDF of $Y=\ln X$ are derived as
%
% EQUAITON #7
\small
\begin{equation}\label{eq:7}
g_Y(y) = \frac{1}{\pi}e^y\sqrt{\frac{1}{e^y} - \frac{1}{4}},      \quad \quad  \text{for } -\infty < y \leq \ln 4
\end{equation}
\normalsize
%and its CDF becomes
%
% EQUAITON #8
\small
\begin{equation}\label{eq:8}
G_Y(y) = \frac{1}{2\pi}\Bigg\{\pi+ e^y\sqrt{4e^{-y}-1} +2\tan^{-1}\Bigg(\frac{e^{-y}(e^y -2)}{\sqrt{4e^{-y}-1}}\Bigg)\Bigg\}.
\end{equation}
\normalsize
%================================================================================%
%We use these PDFs and CDFs to derive the upper and lower bounds of the ergodic SE in~\eqref{eq:4}.

%Since the ICI in~\eqref{eq:3} is dependent of $k$ and the sum in~\eqref{eq:4} is over $m$ instead of $N$, it is not possible to convert the sum in~\eqref{eq:4} to the integral with $f_X(x)$ over the range of $x$ as $N\to \infty$. 
%To resolve this problem, 
We derive Proposition 1 using~\eqref{eq:5} and~\eqref{eq:6} to approximate the expectation of each eigenvalue $\sigma^2_i$. 
Proposition 1 is used to derive the upper and lower bounds of the ergodic SE.
% Proposition #1
\begin{proposition}
The expectation of the $i$-th eigenvalue $\sigma^2_i$ of an \small$N\times N$\normalsize~matrix, whose entries are IID zero-mean complex random variables with unit variance, is approximated by
\vspace{-0.5em}

% EQUAITON #9
\footnotesize
\begin{equation}\label{eq:9}
\begin{aligned}
\mathbb{E}\Big[\sigma^2_i\Big] \simeq  &\frac{N^2}{4\pi}\left[-a_i\,(a_i-2)\sqrt{\frac{4}{a_i}-1}+4\tan^{-1}\left(\frac{(a_i-2)\sqrt{\frac{4}{a_i}-1}}{a_i-4}\right)\right.\\
&\left.+a_{i-1}\,(a_{i-1}-2)\sqrt{\frac{4}{a_{i-1}}-1} +4\tan^{-1}\left(\frac{a_{i-1}-2}{a_{i-1}\sqrt{\frac{4}{a_{i-1}}-1}}\right)\right]
\end{aligned}
\end{equation}
\normalsize 
where $a_i = F^{-1}_X(1-i/N)$, and $\sigma_1 \geq \sigma_2 \geq \cdots \geq \sigma_N$.
\end{proposition}
\begin{proof}
See Appendix A.
\end{proof}
Using Proposition 1, we can approximate the expectation of $\sigma_i^2$ only with the parameters $i$ and $N$. The numerical validity of Proposition 1 is demonstrated in Section~\ref{sec:Num}.

%================================================================================%
%================================================================================%

\subsection{Upper and Lower Bounds of Ergodic Spectral Efficiency}
Lemma 1 in~\cite{hamdi2010useful} is used to convert the ergodic SE to an integral form. 
%
% Lemma #1
%\begin{lemma}[\cite{hamdi2010useful}, Lemma 1]Let $x_1,\cdots,x_N,y_1,\cdots,y_M$ be arbitrary non-negative random variables. Then
%
% EQUAITON #16
%\begin{equation}\label{eq:16}
%\begin{split}
%\mathbb{E}\Bigg[\ln\Bigg(1+&\frac{\sum_{n=1}^{N}x_n}{\sum_{m=1}^{M}y_m+1}\Bigg)\Bigg]\\
%&=\int_0^\infty \frac{\mathcal M_y(z) - \mathcal M_{x,y}(z)}{z}\exp(-z)\mathrm{d}z
%\end{split}
%\end{equation}
%where $\mathcal M_y(z) = \mathbb{E}\big[e^{-z\sum_{m=1}^{M}y_m}\big]$ and $\mathcal M_{x,y}(z) = \mathbb{E}\big[e^{-z(\sum_{n=1}^{N}x_n+\sum_{m=1}^{M}y_m)}\big]$.
%\end{lemma}
%\begin{proof}
%See Lemma 1 in~\cite{hamdi2010useful}.
%\end{proof}
Leveraging Proposition 1 and Lemma 1, we derive Theorem 1 which is the main result in this paper.
% We state our main result for the ergodic SE considering the multiplexing gain~\eqref{eq:4}.
%%==========================================================================%%%
% THEOREM #1
\begin{theorem} 
The ergodic spectral efficiency of a typical mobile user in the MIMO cellular model~\eqref{eq:1} is bounded by
% EQUATION #17
\begin{equation}\label{eq:17}
\begin{split}
r(N, m, \alpha)  \leq {\log_2{e}} \int_0^\infty \frac{1}{z}\frac{m-\sum_{k=1}^m e^{-zU_{N,k}}}{{}_2 F_1\left(m,-\frac{2}{\alpha},1-\frac{2}{\alpha},-z\right)}\,\mathrm{d}z
\end{split}
\end{equation}
%
% EQUATION #18
\begin{equation}\label{eq:18}
\begin{split}
r(N, m, \alpha)  \geq \,&{\log_2{e}} \int_0^\infty \frac{1}{z}\frac{m-\sum_{k=1}^m e^{-ze^{L_{N,k}}}}{{}_2 F_1\left(m,-\frac{2}{\alpha},1-\frac{2}{\alpha},-z\right)}\,\mathrm{d}z
\end{split}
\end{equation}
with 
%
% EQUATION
\footnotesize
\begin{equation*}
\begin{aligned}
U_{N,k} =  &\frac{N^2}{4\pi}\left[-a_k(a_k-2)\sqrt{\frac{4}{a_k}-1}+4\tan^{-1}\left(\frac{(a_k-2)\sqrt{\frac{4}{a_k}-1}}{a_k-4}\right)\right.\\
&\left.+a_{k-1}\,(a_{k-1}-2)\sqrt{\frac{4}{a_{k-1}}-1} +4\tan^{-1}\left(\frac{a_{k-1}-2}{a_{k-1}\sqrt{\frac{4}{a_{k-1}}-1}}\right)\right]\\
\end{aligned}
\end{equation*}
\normalsize
where $a_k = F_X^{-1}(1- k/N)$, and
%
% EQUATION
\footnotesize
\begin{equation*}
\begin{aligned}
L_{N,k} &= \frac{N}{\pi} \int_{b_k}^{b_{k-1}} y\,e^y\sqrt{\frac{1}{e^y}-\frac{1}{4}}\,\mathrm{d}y + \ln N
\end{aligned}
\end{equation*}
\normalsize
where $b_k = G_Y^{-1}(1- k/N)$. The function ${}_2 F_1(\cdot,\cdot,\cdot,\cdot)$ is the Gauss-hypergeometric function defined as
%\vspace{-0.5em}
%
% EQUATION
\small
\begin{equation*}
{}_2 F_1(a,b,c,z) =\frac{\Gamma(c)}{\Gamma(b)\Gamma(c-b)}\int_0^1\frac{t^{b-1}(1-t)^{c-b-1}}{(1-tz)^a}\mathrm{d}t.
\end{equation*}
\end{theorem}
\normalsize

\begin{proof}
See Appendix B
\end{proof}
{ Under assumptions of a Poisson network model, equal power allocation and no ICI cancellation, the derived upper bound shows the best possible system-level performance of any MIMO non-cooperative cellular network since it is the characterization of the optimal MIMO transceiver technique with respect to ergodic SE.}

The proposed bounds provide insight into the optimal number of streams.
The ICI term, ${}_2 F_1\left(m,-\frac{2}{\alpha},1-\frac{2}{\alpha},-z\right)$, increases as the number of streams $m$ increases, so the individual SIR of each stream decreases. 
%which results in the decrease of the SIR, 
%while the increase of $m$ produces the more multiplexing gain. 
Also, the multiplexing gain increases as $m$ increases.
This implies a trade-off between multiplexing gain and SIR gain with respect to $m$. Thus, it is expected that there exists the optimal number of streams $m^*$ that maximizes the ergodic SE.
In the next section, this intuition is validated.
% given the fixed number of antennas $N$. 
%================================================================================%

% FIGURE #2
%\begin{figure}[t]
%\centering
%$\begin{array}{c}
%{\resizebox{0.95\columnwidth}{!}
%{\includegraphics{evalues4.png}}} \\
%\mbox{(a)}\\
%{\resizebox{0.95\columnwidth}{!}
%{\includegraphics{lnevalues4.png}}}  \\  
%\mbox{(b)}
%\end{array}$
%\caption{Comparison of simulation and approximation results for the expectations of (a) the normalized eigenvalues $\lambda^2_k/N$ and (b) $\ln (\lambda^2_k/N)$ for the number of antennas $N\in \{4, 8, 16, 32\}$.} \label{fig:2}
%\end{figure} 

%FIGURE #1
\begin{figure}[t]\centering
\includegraphics[scale = 0.35,trim=4 4 4 4,clip]{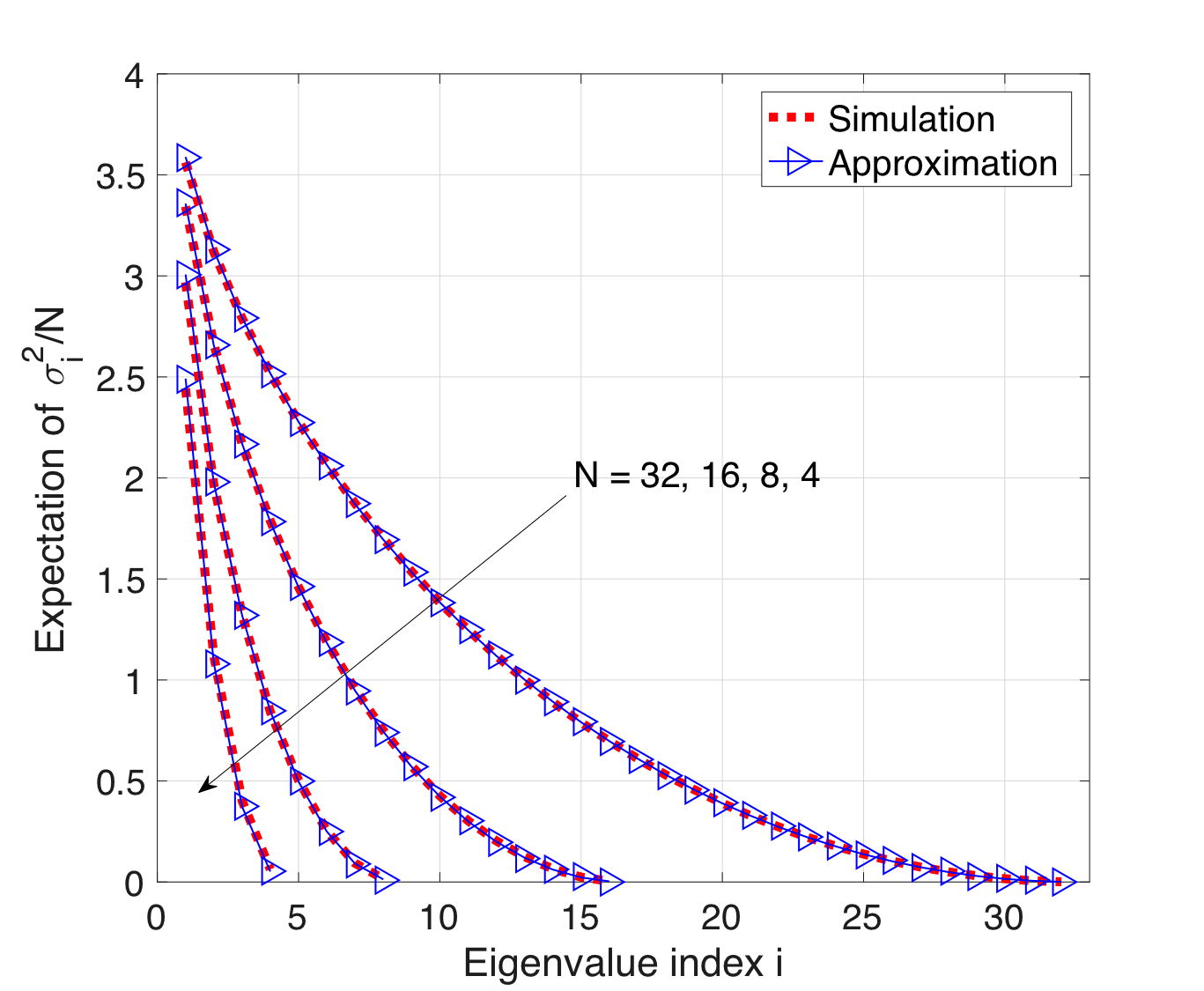}
\caption{The simulation and approximation results of the expectations of $\sigma^2_i/N$ for the number of antennas $N\in \{4, 8, 16, 32\}$.} \label{fig:1}
\end{figure}

% FIGURE #3
%\begin{figure}[t]\centering
%\includegraphics[scale = 0.38]{lnevalues4.png}
%\caption{Comparison for the expectations of $\ln (\lambda^2_k/N)$ between simulation and approximation results with $N = 4, 8, 16$ and $32$.} \label{fig:3}
%\end{figure} 

% Write insights for the optimal number of symbols from the Theorem
\section{Numerical Results}\label{sec:Num}
%We have characterized an upper and lower bounds of the ergodic SE by using the approximation of the eigenvalue expectations.
In this section, we provide the simulation results. 
Fig.~\ref{fig:1} shows the expectation of $\sigma^2_i/N$ for an $N\times N$ channel matrix with $N\in \{4, 8, 16, 32\}$. 
It illustrates that the approximations closely match with the simulation results even for small $N$. 

The obtained bounds and our intuition regarding the optimal number of streams $m^*$ are also validated.
Fig.~\ref{fig:2}(a) shows analytical bounds and simulation results for a pathloss exponent $\alpha \in \{3, 4, 5\}$ and the number of antennas $N \in \{2, 4,\cdots, 16\}$ with $ m = N$. 
%it reduces ICI which comes from other BSs located further than the associated BS for the typical user. 
In the figure, gaps are reasonably tight;
%the proposed bounds provide the tight upper and lower bounds on the ergodic SE.
the average gap between the bounds is about $0.24$ bps/Hz in Fig.~\ref{fig:2}(a). 
Hence, the obtained bounds can estimate the true ergodic SE within the error of $0.24$ bps/Hz.
It is noticeable that the ergodic SE scales almost  linearly with $N$ due to the multiplexing gain and {the gap between the bounds narrows with increasing $N\in \{2, 4,\cdots, 16\}$. }
%so that the bounds are expected to converge as $N\to \infty$.
The ergodic SE also increases as the pathloss exponent $\alpha$ increases since the ICI diminishes faster than the desired signal with the larger $\alpha$.
%This observation corresponds to the insight from Theorem 1.
%Hence, without extensive Monte Carlo simulation, the obtained bounds can estimate the true ergodic SE within the error of $0.24$ bps/Hz.
% for different pathloss exponents and the number of antennas.

\begin{figure}[t]
\centering
$\begin{array}{c}
\mbox{(a)} \\  
{\resizebox{0.95\columnwidth}{!}
{\includegraphics[trim=4 4 4 4,clip]{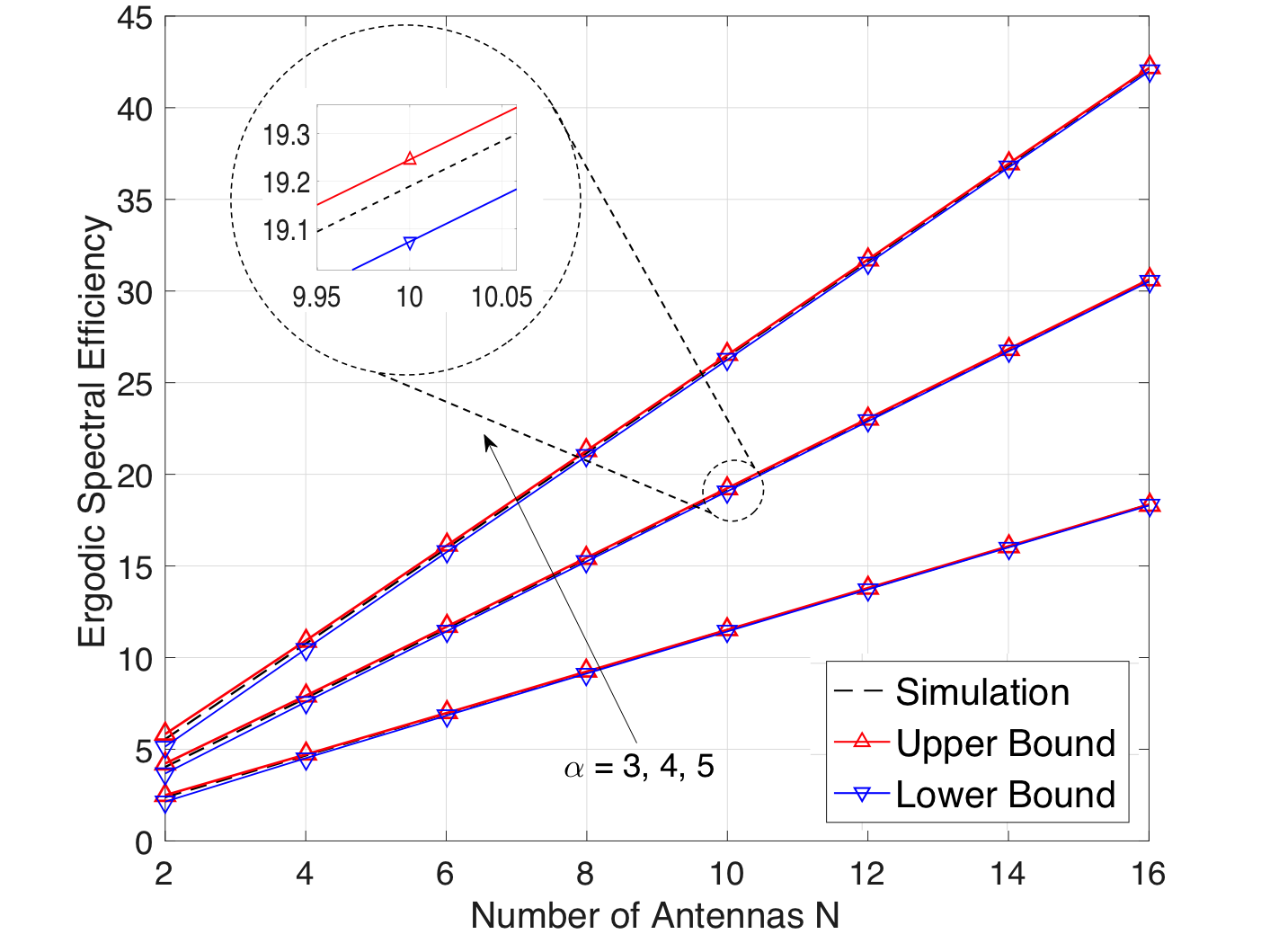}}}  \\
\mbox{(b)}\\
{\resizebox{0.95\columnwidth}{!}
{\includegraphics[trim=4 4 4 4,clip]{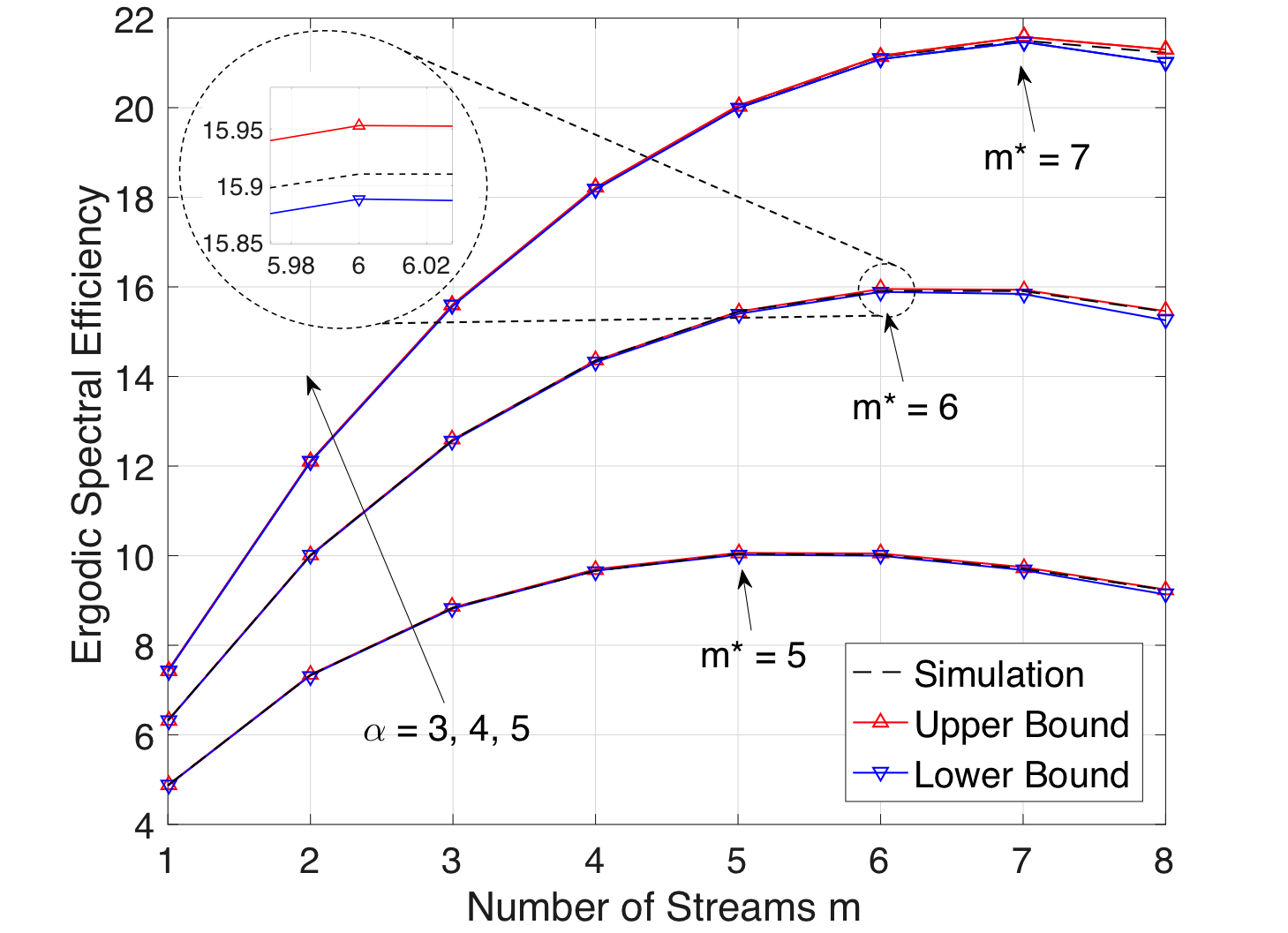}}}  
\end{array}$
\caption{
The ergodic SE of the analytical bounds and simulation results with (a) the different number of antennas $N$ and the $N$ number of streams $m = N$, and (b) the different number of streams $m$ and $N = 8$ antennas for a pathloss exponent $\alpha \in \{3, 4, 5\}$. Here, $m^{*}$ represents the optimal number of streams for simulation results.
} 
\label{fig:2}
\vspace{-1 em}
\end{figure}

%\vspace{-0.3em}
%\subsection{Bounds with the Different Number of Streams }

%We also evaluate the analytic bounds by changing the number of streams $m$.  
% by comparing to the simulations. 
%with changing the number of streams $m$ and compare with simulation results to validate the effectiveness of the bounds with respect to $m$. 
Fig.~\ref{fig:2}(b) shows the analytical bounds and simulation results with $N = 8$ and the different number of streams $m$ for $\alpha \in \{3, 4, 5\}$.
% depending on the number of the transmitted streams $m$.
As observed, the bounds closely match to the simulation results. 
The curves also confirms the intuition that there exists the optimal point of $m$, beyond which sending more streams degrades the ergodic SE.
%This is because 
Specifically, it is more efficient not to use an additional stream with the small eigenvalues due to the trade-off between multiplexing gain and SIR gain as explained from Theorem 1.
%more power to each stream than to obtain the multiplexing gain from small eigenvalues beyond the point.
Moreover, the optimal $m$ is proportional to the pathloss exponent $\alpha$; 
the ICI term ${}_2 F_1\left(m,-\frac{2}{\alpha},1-\frac{2}{\alpha},-z\right)$ in the derived bounds decreases as the pathloss exponent $\alpha$ increases, resulting the increase of the SIR. %; if the pahtloss exponent $\alpha$ increases, the ICI decreases faster than the desired signal. 
Thereupon, as $\alpha$ becomes larger, the multiplexing gain becomes more desirable than the power gain in attempt to maximize the ergodic SE.
This leads to the increase of $m^*$.
% so that the optimal $m$ increases. 
%which corresponds to our insight from Theorem 1; the multiplexing gain becomes more important than the power gain in an attempt to improve the SE as the pathloss exponent $\alpha$ increases. 
%As the pathloss exponent increases, the interference coming from far distance diminishes quickly; thereby the operating SIR increases.
%with arrows indicating the optimal number of streams on simulation results $m^*$ for $$. $m^*$ decreases as $\alpha$ becomes smaller, which is equivalent to the increase of ICI. This shows the trade-off between the multiplexing gain and SNR; as in a common belief, when interference is harsh, 
% the multiplexing gain from small eigenmodes of the channel becomes less beneficial than SNR increase as a result of sending less streams . 

%%%%%% TEST
%Note that the obtained bounds in all three setups provide the close-fitting bounds on the ergodic SE. Hence, 
%Utilizing the tight bounds, we can estimate the optimal number of streams. 
%For instance, $m^{*}$ in Fig.~\ref{fig:2} (b) represents the optimal number of streams from simulation results, and it corresponds to the number of streams that maximizes the ergodic SE from the analytic bounds. 
%Thereby, the optimal streams can be estimated in a system-level view by using the derived bounds.
{ICI cancellation using the remaining $N-m$ receive antennas \cite{veetil2013coverage,veetil2015performance} always decreases interference power, which results in the increase of the SIR.
Consequently, employing the interference cancellation would lightly lift the ergodic SE curves in Fig. \ref{fig:2}(b) except that ergodic SE would remain unchanged at $m = N$, and hence there still exists an optimal value for $m$. }
%Similarly, power allocation by using water-filling algorithm would give marginal lift to the overall trend of the SE.

%we numerically evaluate the optimal number of streams for $N=16$, and Fig.~\ref{fig:5} shows the estimated optimal number of streams $m^{**}$. 
%It also increases as $\alpha$ increases as expected.

\section{Conclusion}

For a point-to-point SVD-MIMO downlink system, this paper derives upper and lower bounds on ergodic spectral efficiency by modeling a cellular network as a homogeneous Poisson point process and approximating the expectation of the channel eigenvalues. 
{ The upper bound on ergodic spectral efficiency applies to any MIMO strategy in a non-cooperative cellular network with equal power allocation across streams and no interference cancellation. }
We conjecture that there exists an optimal number of streams that maximizes the ergodic spectral efficiency, which is proportional to the pathloss exponent.
{ Incorporating interference cancellation in SVD-MIMO analysis would be desirable for future work.}

\begin{appendices}
  \section{Proof of Proposition 1}
Let $X$ be the eigenvalue of $\mathbf{H}/\sqrt{N}$~\cite{tse2005fundamentals} with the PDF of \eqref{eq:5}.
Since singular values of an IID sub-Gaussian square matrix can take a finite interval of values~\cite{wei2016upper} and the empirical distribution of $\sqrt{X}$ tends to quarter circular distribution~\cite{bordenave2012around}, we divide the PDF in \eqref{eq:5} into $N$ non-overlapping regions $\mathcal{A}_i$, $i = 1,\cdots, N$, where each region has a probability of $1/N$.
%Exploiting that the empirical distribution of $\sqrt{X}$ tends to the quarter circular distribution~\cite{bordenave2012around}, 
Considering that $X_i  = \sigma^2_i/N$ has the domain of $\mathcal{A}_i = (a_i,a_{i-1}]$ with $a_i = F_X^{-1}(1-i/N)$, $\mathbb{E}\big[X_i\big]$ can be approximated as
%$\mathcal{A}_i$ is the $i$th domain of the probability of $1/N$ from the probability of $1$ for \eqref{eq:5} 
%The $i$th largest eigenvalue is mostly drawn in the $i$th domain of the probability $1/N$ from the probability of $1$ assuming a limited eigenvalue distribution.
% i.e., $A_i$ is the $i$th range of the probability $1/N$ from the probability of $1$, which we consider as the effective range for $\lambda^2_i/N$. 
%Denoting $X_i = \sigma^2_i/N$, we have
%Let $\lambda^{(n)}$ ($n = 1,2,\cdots,N$) be an unordered singular value, $x^{(n)} = (\lambda^{(n)})^2/N$ and $x_i =\lambda^2_i/N$ where $\lambda_i$ is the $i$th largest singular value. Then, the expectation of $x_i$ can be expressed as
%
% EQUAITON #10
\vspace{-0.5 em}
\begin{equation}\label{eq:10}
\begin{split}
\mathbb{E}\big[X_{i}\big] &\simeq \mathbb{E}\big[X|X \in \mathcal{A}_i\big]= \int_{0}^4 x\,f_X(x|X \in \mathcal{A}_i) \,\mathrm{d}x.
\end{split}
\end{equation} 
The PDF of $X$ given $X \in \mathcal{A}_i$ is represented as
%
% EQUAITON #11
\small
\begin{equation}\label{eq:11}
\begin{split}
\lim_{\Delta x\to 0}f_X(x|X \in \mathcal{A}_i)\Delta x &=\lim_{\Delta x\to 0}\frac{\Pr(X\in [x,x+\Delta x], X\in \mathcal{A}_i)}{\Pr(X\in \mathcal{A}_i)}.
%& \simeq N\cdot \frac{1}{\pi}\sqrt{\frac{1}{x} - \frac{1}{4}} 
\end{split}
\end{equation}
\normalsize
By the definition of $\mathcal{A}_i$, $\Pr(X \in \mathcal{A}_i) = 1/N$ and
\vspace{-0.5em}
% EQUAITON #12
\begin{equation}\label{eq:12}
\Pr(X\in [x, x+\Delta x], X\in \mathcal{A}_i) = \int_x^{x+\Delta x} f^{(i)}_X(x)\mathrm{d}x
\end{equation}
where $f^{(i)}_X(x) = f_X(x)\text{ if } x \in \mathcal{A}_i$, and $f^{(i)}_X(x) =0$ otherwise.
%\begin{equation*}
% f^{(i)}_X(x) =
%\begin{cases}
%    f_X(x)   & \,\text{if } x \in A_i\\
%    0  & \,\text{else.}\\
%\end{cases}
%\end{equation*}
Hence, \eqref{eq:11} becomes
\vspace{-0.5em}
% EQUAITON #14
\small
\begin{equation}\label{eq:14}
\begin{split}
f_X(x|X \in \mathcal{A}_i)&=\lim_{\Delta x\to 0}\frac{\int_x^{x+\Delta x} f^{(i)}_X(x)\mathrm{d}x}{\Delta x/N}= Nf^{(i)}_X(x).
%&= \lim_{\Delta x\to 0} N\frac{F^{(i)}_X(x+\Delta x) - F^{(i)}_X(x)}{\Delta x}\\
%&= Nf^{(i)}_X(x).
\end{split}
\end{equation}
\normalsize
Finally, we put~\eqref{eq:14} into~\eqref{eq:10} with $\mathbb{E}[\sigma^2_i] = N\,\mathbb{E}[X_i]$.
% completes the proof. 
%this gives the conditional expectation~\eqref{eq:10} as
%
% EQUAITON #15
%\footnotesize
%\begin{equation}\label{eq:15}
%\begin{split}
%\mathbb{E}[X_{i}]& {\simeq}\int_{a_{i}}^{a_{i-1}} x\,\frac{N}{\pi}\sqrt{\frac{1}{x}-\frac{1}{4}} \,\mathrm{d}x\\
%&=\frac{N}{4\pi}\left[-a_i\,(a_i-2)\sqrt{\frac{4}{a_i}-1}+4\tan^{-1}\left(\frac{(a_i-2)\sqrt{\frac{4}{a_i}-1}}{a_i-4}\right)\right.\\
%&\left.+a_{i-1}\,(a_{i-1}-2)\sqrt{\frac{4}{a_{i-1}}-1} +4\tan^{-1}\left(\frac{a_{i-1}-2}{a_{i-1}\sqrt{\frac{4}{a_{i-1}}-1}}\right)\right]
%\end{split}
%\end{equation}
%\normalsize
%Recall that $X_i=\lambda^2_i/N$, so we have $\mathbb{E}[\lambda^2_i] = N\,\mathbb{E}[X_i]$.
\qed
\vspace{-0.5em}

\section{Proof of Theorem 1}

(\textit{Upper bound}) Replacing SIR${}_k$ in~\eqref{eq:4} with~\eqref{eq:3}, we have the ergodic SE expressed as
\vspace{-0.5em}

%EQUATION #19
\small
\begin{equation}\label{eq:19}
\begin{split}
r(N, m, \alpha, \lambda) &= \mathbb{E}\Bigg[ \sum_{k=1}^{m}{\log_2\left(1+\frac{\|\mathbf{d}_{1}\|^{-\alpha}{\sigma}^{2}_{k}}{\sum_{i=2}^{\infty}{\|\mathbf{d}_{i}\|^{-\alpha}{q}_{i,k}}}\right)}\Bigg]\\
&\stackrel{(a)}{\leq}\mathbb{E}\Bigg[\sum_{k=1}^{m}\log_2\left(1+\frac{\|\mathbf{d}_1\|^{-\alpha}\mathbb{E}\left[\sigma^2_k\right]}{\sum_{i=2}^{\infty}{\|\mathbf{d}_{i}\|^{-\alpha}{q}_{i,k}}}\right)\Bigg]\\
&\stackrel{(b)}{=}\log_2{e} \sum_{k=1}^{m}\int_0^{\infty} \frac{1}{z}\left(1-\mathcal{M}_{s_u}(z)\right)\mathcal{M}_I(z)\mathrm{d}z
\end{split}
\end{equation}
\normalsize
%\vspace{-0.3em}
with
% EQUATION
\footnotesize
\begin{equation}\label{eq:lemma}
% \mathcal{M}_{s_u}(z) = \mathbb{E}\left[e^{-z\mathbb{E}\left[\sigma^2_k\right]}\right]\small\text{, } \small\mathcal{M}_I(z) = \mathbb{E}\left[e^{-z\|\mathbf{d}_1\|^\alpha \sum_{i=2}^\infty{\|\mathbf{d}_{i}\|^{-\alpha}{q}_{i,k}}}\right]
 \mathcal{M}_{s_u}(z) = \mathbb{E}\left[e^{-z\mathbb{E}\left[\sigma^2_k\right]}\right]\small\text{, } \small\mathcal{M}_I(z) = \mathbb{E}\left[e^{-z\|\mathbf{d}_1\|^\alpha \sum_{i=2}^\infty{\frac{{q}_{i,k}}{\|\mathbf{d}_{i}\|^{\alpha}}}}\right]
\end{equation}
\normalsize
where (a) is from Jensen's inequality and (b) comes from Lemma 1 in~\cite{hamdi2010useful}. Under Proposition 1, $\mathcal{M}_{s_u}(z)$ becomes
$\mathcal{M}_{s_u}(z)=\mathbb{E}\left[e^{-z\mathbb{E}\left[\sigma^2_k\right]}\right] \simeq e^{-zU_{N,k}}$,
%\vspace{-0.3em}
%EQUATION #20
%\small
%\begin{equation}\label{eq:20}
%\mathcal{M}_s(z)=\mathbb{E}\left[e^{-z\mathbb{E}\left[\lambda^2_k\right]}\right] \simeq e^{-zU_{N,k}} 
%\end{equation}
%\normalsize
where $U_{N,k}$ is defined in Theorem 1.
%
%EQUATION
%\footnotesize
%\begin{equation*}
%\begin{aligned}
%U_{N,k} &=  \frac{N^2}{4\pi}\left[-a_k(a_k-2)\sqrt{\frac{4}{a_k}-1}+4\tan^{-1}\left(\frac{(a_k-2)\sqrt{\frac{4}{a_k}-1}}{a_k-4}\right)\right.\\
%&\left.+a_{k-1}(a_{k-1}-2)\sqrt{\frac{4}{a_{k-1}}-1} +4\tan^{-1}\left(\frac{a_{k-1}-2}{a_{k-1}\sqrt{\frac{4}{a_{k-1}}-1}}\right)\right]
%\end{aligned}
%\end{equation*}
%\normalsize
The Laplace transform of the ICI $\mathcal{M}_I(z)$ is derived by closely following the Appendix D in~\cite{park2015optimal} as
$\mathcal{M}_I(z)=1/{}_2F_1\left(m,-\frac{2}{\alpha},1-\frac{2}{\alpha}, -z\right)$.
%\vspace{-0.5em}
%EQUATION #21
%\small
%\begin{equation}\label{eq:21}
%\mathcal{M}_I(z)={}_2F_1^{-1}\left(m,-\frac{2}{\alpha},1-\frac{2}{\alpha}, -z\right).
%\end{equation}
%\normalsize
%Note that the parameter $m$ in ${}_2F_1\left(m,-\frac{2}{\alpha},1-\frac{2}{\alpha}, -z\right)$, which represents the number of streams comes from $g_{l,k} \sim \chi^2_{2m} $. 
%Thus, the upper bound~\eqref{eq:19} becomes 
%
%EQUATION #22
%\small
%\begin{equation}\label{eq:22}
%\begin{split}
%r(N, m, \alpha,\delta)&\,
%{\leq}\log_2{e} \sum_{k=1}^{m}\int_0^{\infty} \frac{1}{z}\left(1-\mathcal{M}_s(z)\right)\mathcal{M}_I(z)\mathrm{d}z\\
%&\leq {\log_2{e}}\sum_{k=1}^m \int_0^\infty \frac{1}{z}\frac{1-e^{-zU_{N,k}}}{{}_2 F_1\left(m,-\frac{2}{\alpha},1-\frac{2}{\alpha},-z\right)}\,\mathrm{d}z
%&\cdot \frac{1}{{}_2 F_1\big(m,-\frac{2}{\alpha},1-\frac{2}{\alpha},-z\big)}\,\mathrm{d}z
%\end{split}
%\end{equation}
%\normalsize
This completes the proof for the upper bound.

(\textit{Lower bound}) 
With $\mathcal{M}_{s_l}(z) = \mathbb{E}\left[e^{-ze^{\mathbb{E}\left[\ln\sigma^2_k\right]}}\right]$ and $\mathcal{M}_I(z)$ in~\eqref{eq:lemma}, the ergodic SE~\eqref{eq:4} is lower bounded by 
\vspace{-0.5em}

% EQUATION #23
\small
\begin{align*}\label{eq:23}
r(N, m, \alpha,\lambda) 
%&= \mathbb{E}\Bigg[ \sum_{k=1}^{m}{\log_2\left(1+\frac{r^{-\alpha}{\lambda}^{2}_{k}}{\sum_{l\in \Phi\setminus b_{0}}{R_{l}^{-\alpha}{g}_{l,k}}}\right)}\Bigg]\\
&\stackrel{(c)}{\geq}\mathbb{E}\Bigg [ \sum_{k=1}^{m}\log_2\left(1+\frac{\|\mathbf{d}_1\|^{-\alpha}e^{\mathbb{E}\left[\ln\sigma^2_k\right]}}{\sum_{i=2}^\infty{\|\mathbf{d}_{i}\|^{-\alpha}{q}_{i,k}}}\right)\Bigg ]\\
&\stackrel{(d)}{=}\log_2{e} \sum_{k=1}^{m}\int_0^{\infty} \frac{1}{z}\left(1-\mathcal{M}_{s_l}(z)\right)\mathcal{M}_I(z)\mathrm{d}z
\end{align*}
\normalsize
% EQUATION
%= \mathbb{E}\left[e^{-z\|\mathbf{d}_1\|^\alpha \sum_{l\in \Phi\setminus b_{0}}{R_{l}^{-\alpha}{q}_{l,k}}}\right]
The inequality (c) is from Jensen's inequality and (d) is from Lemma 1 in~\cite{hamdi2010useful}. The expectation $\mathbb{E}\left[\ln\sigma^2_k\right]$ can be derived by following the similar steps in the proof of Proposition 1; 
let $Y_k=\ln(\sigma^2_k/N)$, then $\mathbb{E}[Y_k]$ can be approximated by using~\eqref{eq:7} and~\eqref{eq:8}, instead of~\eqref{eq:5} and~\eqref{eq:6} as 
$\mathbb{E}[Y_k] \simeq \frac{N}{\pi} \int_{b_k}^{b_{k-1}} y\,e^y\sqrt{\frac{1}{e^y}-\frac{1}{4}}\,\mathrm{d}y
$,
%EQUATION #24
%\small
%\begin{equation}\label{eq:24}
%\begin{aligned}
%\mathbb{E}[Y_k] \simeq \frac{N}{\pi} \int_{b_k}^{b_{k-1}} y\,e^y\sqrt{\frac{1}{e^y}-\frac{1}{4}}\,\mathrm{d}y
%\end{aligned}
%\end{equation}
%\normalsize
where $b_k = G_y^{-1}(1- k/N)$.
Since $\mathbb{E}[\ln\sigma^2_k] = \mathbb{E}[Y_k] +\ln N$,  $\mathcal{M}_{s_l}(z)$ becomes $\mathcal{M}_{s_l}(z) \simeq e^{-ze^{L_{N,k}}}$,
%\begin{align}
%\mathcal{M}_s(z) = \mathbb{E}\left[e^{-ze^{\mathbb{E}\left[\ln\lambda^2_k\right]}}\right] \simeq e^{-ze^{L_{N,k}}}
%\end{align}
where $L_{N,k}$ is defined in Theorem 1. We omit $\lambda$ in the ergodic SE $r(\cdot)$ as the derived bounds are not a function of $\lambda$.
\qed
\end{appendices}

%\vspace{-0.5em}
\bibliographystyle{IEEEtran}
\bibliography{WCL16.bib}
\end{document}